\documentclass{amsart}
\usepackage[latin9]{inputenc}
\usepackage{geometry}
\usepackage{verbatim}
\usepackage{amstext}
\usepackage{amsthm}
\usepackage{amssymb}
\usepackage{graphicx}
\usepackage[hidelinks=true]{hyperref}
\usepackage{mathabx}

\makeatletter

\newcommand{\ep}{\eta^{\uparrow}}

\newcommand{\R}{ {\mathbb{R}} }
\newcommand{\E}{ {\mathbb{E}} }

\newtheorem{theorem}{Theorem}\theoremstyle{plain}

\renewcommand{\and}{\quad\textrm{ and }\quad}

\newcommand{\cE}{\mathcal{E}}
\newcommand{\ccE}{\check{\cE}}
\newcommand{\cI}{\mathcal{I}} 
\newcommand{\ccI}{\check{\mathcal{I}}} 
\newcommand{\cN}{\mathcal{N}}

\newcommand{\ce}{\epsilon}

\usepackage{xcolor}

\newtheorem{assumption}{Assumption}\newtheorem{lemma}{Lemma}\newtheorem{proposition}{Proposition}\newtheorem{remark}{Remark}\numberwithin{equation}{section}

\numberwithin{equation}{section}
\numberwithin{lemma}{section}
\numberwithin{proposition}{section}
\numberwithin{theorem}{section}
\numberwithin{corollary}{section}
\numberwithin{definition}{section}
\numberwithin{remark}{section}

\title{Weak error rates of numerical schemes for rough volatility}

\author{Paul Gassiat}
\address{
Universit\'e Paris-Dauphine, PSL University, UMR 7534, CNRS, CEREMADE, 75016 Paris, France}
\email{gassiat@ceremade.dauphine.fr}

\thanks{This work is partially supported by the ANR via the project ANR-16-CE40-0020-01. The author is indebted to two anonymous referees whose comments helped to substantially improve the presentation.}

\begin{document}

\begin{abstract}
Simulation of rough volatility models involves discretization of stochastic integrals where the integrand is a function of a (correlated) fractional Brownian motion of Hurst index $H \in (0,1/2)$. We obtain results on the rate of convergence in the number of time-steps for the weak error of such approximations, in the special cases when either the integrand is the fBm itself, or the test function is cubic. Our result states that the convergence is of order $(3H+ \frac{1}{2}) \wedge 1$ for exact left-point discretization, and of order $H+\frac{1}{2}$ for the hybrid scheme with well-chosen weights. 
\end{abstract}

\thanks{ }

\maketitle


\section{Introduction}

The family of rough volatility models, where the volatility process has sample paths which are rougher than those of classical Brownian motion, has been the object of much interest in the mathematical finance community in the last few years, due to its ability to reproduce several features of asset prices, such as for instance the observed skew of implied volatility \cite{ALV07,fukasawa2011asymptotic}, its consistency with empirical time series \cite{GJR18} and the fact that it arises as scaling limit of microstructure models under natural conditions \cite{el2018microstructural, JR20}.

In their simplest forms, the volatility process $\sigma_t$ is a function of a (Riemann-Liouville) fractional Brownian motion (fBm), namely
\[ \sigma_t = f(t, \widehat{W}_t) , \;\;\; \widehat{W}_t = \int_0^t (t-s)^{H-1/2} dW_s,\]
and the corresponding stock-price dynamics are given by
\[ dS_t = \sigma_t (\rho dW_t + \sqrt{1-\rho^2} d\bar{W}_t). \]
Here $f$ is a deterministic function (the popular ''rough Bergomi'' model, introduced in \cite{BFG16}, corresponds to $f(t,x) = \zeta(t) \exp(\eta x)$), and $W, \bar{W}$ are independent Brownian motions.


By a classical conditioning argument (the so-called Romano-Touzi formula \cite{RT97}), when evaluating European options, the dependence in $\bar{W}$ in $S$ can be integrated out, so that call option prices in this model are given by
\[  \E \left[ C_{BS} \left(  S_{0}\exp \left( \rho \int_{0}^{T} f(t,\widehat{W}_t) dW_t-\frac{\rho ^{2}}{2} 
\int_{0}^{T} f(t,\widehat{W}_t)^2 dt\right) ,K,\frac{1-{\rho}^{2}}{2}\int_{0}^{T} f(t,\widehat{W}_t)^2 dt \right)\right]\]
where $C_{BS}=C_{BS}(S_0,K,\sigma^2 T)$ is the  usual Black-Scholes Call pricing function. Since, unlike the case of Markovian models, PDE methods are not available here, in order to evaluate option prices we are left with the issue of simulating the inner random variable, and, in particular, the stochastic integral
\[ \cI = \int_0^T f(t, \widehat{W}_t) dW_t.\]

A simple choice consists in left-point discretization of the above, namely to write
\[ \mathcal{I} \;\; \approx \;\; {\cI}^{',N} :=  \sum_{k=0}^{N-1} f(\widehat{W}_{t_k}) \left(W_{t_{k+1}} - W_{t_k} \right).\]
where $\{ t_k, k=0,\ldots, N\}$ is a discretization of $[0,T]$. Since the covariance of the Gaussian vector $\left( \widehat{W}_{t_k}, W_{t_k} \right)_{k=1,\ldots, N}$ is explicit, it can be simulated exactly by the classical Cholesky method.
 
It is then natural to ask what is the error made when considering this approximation. It is important here to distinguish between strong and weak error. The strong error corresponds to the size of the difference $\cI - \cI'$, and a simple computation based on It\^o isometry and properties of the fBm show that its $L^2$ norm is of order $N^{-H}$. Since in practical applications $H$ is small (of order $0.1$), this converges to $0$ very slowly which may lead to doubt the practical feasibility of Monte Carlo approximations for these models. However, the more relevant quantity in practice is the weak error, i.e. the quantity
\[ \cE_{\Phi} = \E[\Phi(\cI)] - \E[ \Phi(\cI^{',N})] \]
for a given (family of) test function(s) $\Phi$. 

It is well-known that these two errors do not in general share the same order of convergence (recall that in the case of classical SDEs these orders are respectively $\frac{1}{2}$ and $1$, see e.g. \cite{TT90}). This turns out to also be the case here, as proved first in \cite{BHT20}. They show that the rate of weak convergence is of order at least $H+\frac1 2$ when $f(x)=x$, and in fact they give a simple argument (which they attribute to Neuenkirch) showing that the rate is even of order $1$ when $\Phi$ is a quadratic. Note that these rates have the appealing feature of not going to $0$ as $H \to 0$.

The main result of this work (Theorem \ref{thm:main} below) is a further improvement on their result, showing that, when either $f(x)=x$ or $\Phi$ is a cubic polynomial, the weak error is in fact bounded by a higher power of $1/N$, namely
\[  \mathcal{E}_{\Phi} \leq C \left(\frac{1}{N}\right)^{(3H+1/2)\wedge 1 }. \]

Our proof is based on a direct manipulation of fractional integrals and an application of the integration by parts formula of Malliavin calculus, as first used in the context of numerical error study in \cite{CKL06}. We then prove our result by induction on the regularity of the test function (using crucially the rate $1$ for quadratics in the induction step). The method is arguably simpler than the PDE methods of \cite{BHT20} (based on Markovian approximation). We also believe that our proof could be refined to show that the order above is in fact optimal, but we do not pursue this here. We however present some numerical tests which are consistent with this belief. 

Of course our result is only partial, since we do not treat the case where both $f$ and $\Phi$ are arbitrary, which is the relevant case for practical situations. (Note in particular that, in the case $f(x)=x$ that we treat here, there are faster methods than Monte Carlo for option pricing, such as the Fourier inversion techniques described in \cite{AJ20}). It is not clear if the proof below can be extended to this general case, the induction argument relying strongly on the fact that $f$ is linear. Our result is also not directly applicable to option pricing, since the Romano-Touzi formula differs from the expectations we consider on two aspects  : (i) it depends not only on $\mathcal{I}$ but also on the realized variance $\int_0^T f(\widehat{W}_t)^2 dt$ (ii) it involves evaluation of functions which are typically smooth but not with bounded derivatives, unlike what we require here (we leave a rigorous investigation of these technical points to future research). 

Our method of proof is however quite flexible when it comes to the choice of the approximation, which we highlight by considering next a different approximation for $\cI$, namely that coming from the so-called hybrid scheme \cite{BLP15}. Recall that it consists in replacing $\widehat{W}_t$, at a grid-point $t$, by an approximation

\[ \widecheck{W}_t = \int_{t-\kappa T/N}^t (t-s)^{H-1/2} dW_s + \sum_{j=0}^{k-\kappa - 1}  \check{k}_{k-j}  \left(\int_{t_j}^{t_{j+1}} dW_s \right) \]
where the weights $\check{k}_{\ell}$ correspond to approximating the kernel $k:r\mapsto r^{H-1/2}$ by a constant function on the interval $ [\ell T/N, (\ell+1)T/N]$. Any reasonable choice of the weights lead to a strong convergence of order $H$, but interestingly we observe here that for weak convergence the situation is very different. Indeed, considering quadratic $\Phi$, it is clear that many of the usual choices proposed in the literature lead to a weak convergence of order no better than $2H$. However, choosing the $\check{k}_{\ell}$ in order to match second moments of $\widecheck{W}_t$ and $\widehat{W}_t$, we show that (in the same cases as those considered above), the weak error is of order $N^{-H-1/2}$, see Theorem \ref{thm:hybrid} below. We note that this choice of weights had been proposed in \cite{HJM17}, along with the observation that they lead to a reduced error. Our results give a theoretical justification for the use of these weights.

Finally, we mention the related recent preprint \cite{BFN22}, which uses essentially the same method of proof to study the weak error of the Cholesky discretization when $f(x)=x$ (but they only obtain the suboptimal rate of $H+1/2$). They also prove that the weak rate is at least $2H$ for general $f$. Our results were obtained independently.

The organization of the article is as follows. In Section \ref{sec:2}, we consider the weak error of left-point approximation when the discretized process is sampled exactly. In Section \ref{sec:3}, we study the same quantity when this discretization is replaced by its approximation obtained from the hybrid scheme. Some technical proofs are relegated to Section \ref{sec:4}. 

\section{Weak error for exact discretization} \label{sec:2}
Without loss of generality we will consider the time horizon $T=1$ throughout. We consider a scalar Brownian motion $W$ on $[0,1]$, and the associated Riemann-Louville fBm of Hurst index $H \in (0,1/2)$, defined by
\[ \widehat{W}_t = \int_0^t K(t,s) dW_s, \;\;\;K(t,s) = (t-s)_+^{H-1/2}. \]

Given a fixed function $f$, we let 

\[\cI = \int_0^1 f(\widehat{W}_t) d W_t.\]

For a fixed $n$, we consider the discretization of the above It\^o integral
  \[ \cI' = \int_0^1 f(\widehat{W}_{\eta(t)}) dW_t\]
where for $t \in [0,1]$, $\eta(t) = \lfloor nt \rfloor/n $.

Given a test function $\Phi$, we consider the associated weak error
\[
 \mathcal{E}_{\Phi} = \E \left[\Phi \left(\cI \right) \right] - \E \left[\Phi \left(\cI' \right) \right]  .
\]

Our main result is then the following rate of convergence to $0$ of this quantity, in the cases where either $f$ is linear or $\Phi$ is cubic.

\begin{theorem} \label{thm:main}
Assume that $H \neq \frac{1}{6}$\footnote{We leave to the interested reader to check that, in the case $H=\frac{1}{6}$, the same proof gives a rate of $\frac{\log(n)}{n}$.} and either :

(1) $f(x)=x$ and $\Phi$ is a $C^{\left(2 \lceil1/4H \rceil+3\right) \wedge \left(2 \lceil1/2H \rceil+1\right)}_b$ function,

or

(2) $f$ is $C^2_b$ and $\Phi$ is a cubic polynomial.

Then there exists a constant $C$, which does not depend on $n$, such that
\[ \left| \mathcal{E}_{\Phi} \right| \leq C \left(\frac{1}{n}\right)^{(3H+1/2)\wedge 1 }. \]

\end{theorem}

Before proving the theorem, we provide a numerical illustration. In Figure \ref{fig:cubic} below, we plot $\cE_{\Phi}$ for $\Phi(x)=x^3/6$, $f(x)=x$, \footnote{Technically, this choice does not fulfill the assumptions of Theorem \ref{thm:main} since neither $\Phi$ nor $f$ are bounded, but it is straightforward to check that the proofs of both cases (1) and (2) still go through.} and various values of $H$ and $\Delta$. Note that in that case, $\cE_{\Phi}$ can be computed without Monte-Carlo simulations, since (see proof of Theorem \ref{thm:main}, case (2), p.8 below)
\[ \E[\cI^3] = 6 \int_{0\leq s \leq t \leq 1} dt\; ds \;\E[\widehat{W}_t \widehat{W}_s] K(t,s) = 6 \int_{0 \leq r\leq s \leq t \leq 1} dt\; ds \; dr \; K(t,s) K(t,r) K(s,r) \] 
which can be computed numerically, and similarly $\E[(\cI')^3] $ can be written as a sum involving the correlation function of $\widehat{W}$ on grid-points. The plot is consistent with rate $(3H+1/2) \wedge 1$ being optimal.

\begin{figure}[!h]
\includegraphics[width=400bp]{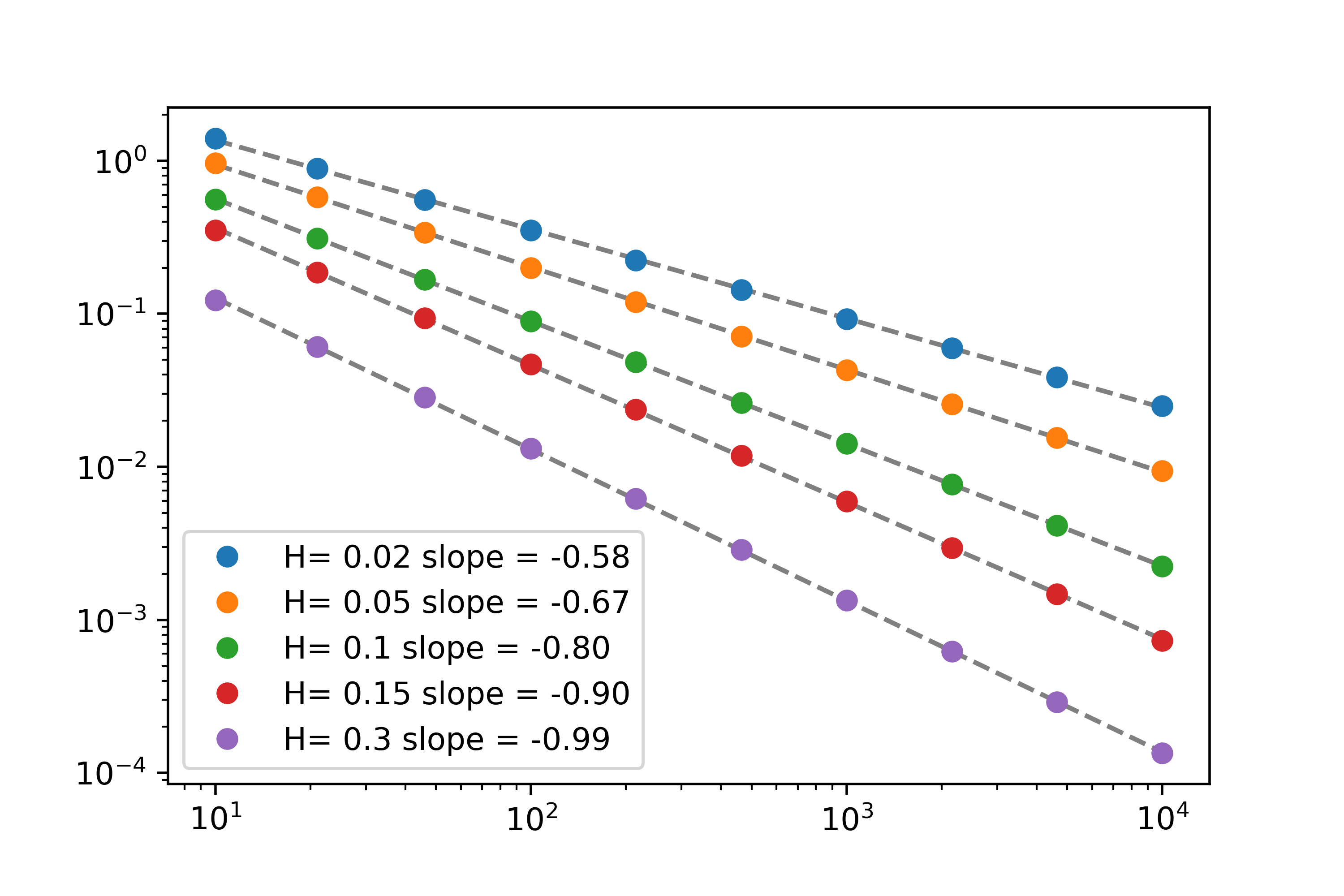}
\caption{Plot of $\cE_{x^3/6}$ as a function of $n$ for various values of $H$, when $f(x)=x$. The dotted lines are linear regressions, with slopes indicated in the legend.}
\label{fig:cubic}
\end{figure}

\subsection{Proof of Theorem \ref{thm:main}, Case (1).\\}

We first introduce some notations.
 Throughout the rest of this paper, we will write $f\lesssim g$ or $f=O(g)$ if $f \leq Cg$ for some constant $C$ that does not depend on $n$.

 Recall that $K(t,s) = (t-s)_+^{H-1/2}$, and we further let $K'(t,s) = K(\eta(t),s)$, $\Delta K = K'-K$, $\Delta(K^2) = (K')^2-K^2$, and for $\theta \in [0,1]$, let $K^\theta =(1- \theta) K +  \theta K'$. Similarly define $\Delta \widehat{W}$ and $\widehat{W}^\theta$.

We start by recording some elementary inequalities on the kernels.

\begin{lemma} \label{lem:inequalities}
It holds that
\begin{equation} \label{eq:strong}
\sup_{t \in[0,1]} \int_0^t ds \; \Delta K(t,s)^2  \lesssim n^{-2H}
\end{equation}

\begin{equation} \label{eq:ineq1}
\mbox{ For any $\alpha \geq 0$, $\alpha \neq H+\frac{1}{2}$,  }\sup_{t\in[0,1]}  \int_0^t  ds\; \left| \Delta K(t,s) \right| (t-s)^{\alpha} \lesssim n^{-H-1/2-\alpha},
\end{equation}

\begin{equation} \label{eq:wrquad}
\forall t \in [0,1],  \left| \int_0^t  ds\; \Delta (K^2)(t,s) \right| \lesssim  n^{-2H} \wedge n^{-1} t^{2H-1} .
\end{equation}

\begin{equation} \label{eq:ineq2}
\forall t \in [0,1], \;\;  \left| \int_0^t  ds \; \Delta K(t,s) \right|  \lesssim n^{-(H+1/2)} \wedge n^{-1} t^{H-1/2}.
\end{equation}

\end{lemma}

\begin{proof}
In order to prove \eqref{eq:strong}, note that for $\eta(t) - s \geq n^{-1}$, one has
\[ \left| K(t,s) - K(\eta(t),s)\right| \leq C \frac{1}{n} (t-s)^{H-3/2}, \]
and split the integral into
\begin{equation*}
\int_{|t-s| \leq 2/n} ds \;(K(t,s)^2 + K(\eta(t), s)^2) + \int_{|t-s| \geq 2/n} ds \;C \frac{1}{n^2} (t-s)^{2H-3}  \lesssim n^{-2H}.
\end{equation*}
\eqref{eq:ineq1} is proved in the same way.

\eqref{eq:wrquad} and \eqref{eq:ineq2} are easy since the inner integrals can be computed exactly, e.g.
\[
 \int_0^t  ds\; (K(t,s)^2 - K'(t,s)^2) = (2H)^{-1} \left(t^{2H} - \eta(t)^{2H}\right) \lesssim  n^{-2H} \wedge n^{-1} t^{2H-1}   
  \]
\end{proof}
Recall that

\[\cI = \int_0^1 \widehat{W}_t d W_t,  \;\;\;\; \cI' = \int_0^1 \widehat{W}_{\eta(t)} dW_t,\]
and let
\[ \Delta \cI = \cI' - \cI  = \int_0^1 \Delta \widehat{W}_t dW_t, \;\;\; \mbox{ and for } \theta \in [0,1], \;\;\cI^\theta = (1-\theta) \cI' + \theta \cI = \int_0^1 \widehat{W}^\theta_t dW_t. \]

We have 
\begin{align*}
\mathcal{E}_{\Phi}=  \int_0^1 d\theta\;  \E \left[ {\Phi}'(\cI^\theta)    \Delta \cI \right].
\end{align*}

We denote by $D$ the Malliavin derivative operator (w.r.t. $W$). Note that $\Delta \cI = \int_0^1  \left(\int_0^t  \Delta K(t,s) dW_s\right) dW_t$, being a double Wiener integral, satisfies $D_{s} D_{t} \Delta \cI = \Delta K(t,s)$ for $s \leq t$.

Applying twice the integration by parts formula from Malliavin calculus, we obtain
\begin{align*}
\mathcal{E}_{\Phi} &= \int_0^1 d\theta  \int_0^1 dt\int_0^t ds   \; \E \left[D_s D_t ({\Phi}'(\cI^\theta))  \right] \Delta K(t,s).
\end{align*} 

Recall that $D_t \left( \int_0^1 u_r dW_r\right) = u_t + \int_t^1 D_t u_r dW_r$, as long as $u$ is an adapted Malliavin differentiable process. This yields 
\[ D_t \cI^\theta =  \widehat{W}^\theta_{t} +  \overline{W}^\theta_t =: \widetilde{W}_t^\theta,\]
where
\[
 \overline{W}^\theta_t := \int_t^1 K^\theta(r,t) dW_r.
 \]
 
We also have for $s \leq t$
\[ D_s D_t \cI^{\theta} = K^\theta(t,s) .\]

Using the chain rule for the Malliavin derivative, this leads to 
\begin{align} \label{eq:EPsi}
\mathcal{E}_{\Phi} &= \int_0^1 d \theta \int_0^1 dt \int_0^t ds\; \E\left[\Phi^{(3)}(\cI^\theta) \widetilde{W}_s^\theta \widetilde{W}^\theta_t\right] \Delta K(t,s) \nonumber \\
&\;\;\;\;\;\; +   \int_0^1 d \theta \int_0^1 dt \int_0^t ds \; \E\left[\Phi^{''}(\cI^\theta)\right] K^\theta(t,s) \Delta K(t,s) 
\end{align}

In order to estimate the first term, we need to study the continuity properties of the expectation appearing in the integral. This is done in the following lemma, the proof of which is a bit tedious and relegated to section \ref{subsec:41}. (Note that when $\Psi \equiv 1$ and $\theta=0$, the considered quantity is simply the correlation function of the Gaussian process with kernel $K(t,s)+K(s,t) = |t-s|^{H-1/2}$, which has similar properties as the correlation function of the fBm, in particular $2H$-H\"older continuity).

\begin{lemma} \label{lem:CorrelPhi}
Given a function $\Psi: \R \to \R$, for any $\theta \in [0,1]$, the map
\begin{equation*}
C_{\Psi}^\theta : (s,t) \mapsto \E\left[\Psi(\cI^\theta)\widetilde{W}^\theta_s \widetilde{W}^\theta_{t} \right]
\end{equation*}
satisfies, in the case where $\Psi$ is bounded,
\begin{equation*}
\left| C_{\Psi}^{\theta}(t,t)\right| \lesssim 1+ \ce(t)^2,
\end{equation*}
and if in addition $\Psi \in C^1_b$,
\begin{equation*}
\left| C_{\Psi}^{\theta}(t,t) - C_{\Psi}^{\theta}(t,s) \right|\lesssim (t-s)^{2H} + n^{-2H} + \ce(t)^2 + \ce(t)\ce(s),
\end{equation*}
where$\ce(t) := n^{-1/2} \left( \lceil nt \rceil/n -t\right)^{H-1/2}$ satisfies
\begin{equation*}
\int_0^1  \ce(t)^2 dt \lesssim n^{-2H},\;\;\; \int_0^1 dt \int_0^t ds\; |\Delta K(t,s)| \ce(t) \ce(s) \lesssim n^{-3H-1/2}.
\end{equation*}
\end{lemma}

\begin{proof}[Proof of Theorem \ref{thm:main} (1)]
We prove by induction on $k$ the slightly more general claim : if $\Phi \in C^{k}_b$, $k$ an odd integer, then, uniformly over $\theta \in [0,1]$, 
\begin{equation}
\E \left[ \Phi\left( \cI^\theta \right)\right] - \E \left[\Phi \left( \cI^{1-\theta} \right)\right] = O\left(n^{-(3H+1/2)\wedge 1}\right) + O\left(n^{-kH}\right). 
\end{equation}
(the result of the theorem corresponds to $\theta=0$, $k H \geq  \left(3H +1/2\right) \wedge 1$).

The case $k=1$ is simple (using that $\Phi$ is Lipschitz and strong error is of order $H$), and we now fix $k\geq 3$.

By the same computation as in \eqref{eq:EPsi}, one has
\begin{align*}
\E \left[ \Phi\left( \cI^\theta \right)\right] - \E \left[\Phi \left( \cI^{1-\theta} \right)\right]   = &\int_\theta^{1-\theta}d\gamma \int_0^1 dt \int_0^t ds\; \Delta K(t,s) \E\left[ \Phi^{(3)}\left( \cI^\gamma \right) 
\widetilde{W}^\gamma_s \widetilde{W}^\gamma_t  \right] \\
&+  \int_\theta^{1-\theta} d\gamma \int_0^1 dt \int_0^t ds\; \Delta K(t,s) K^\gamma(t,s) \E\left[ \Phi''\left(  \cI^\gamma \right)\right] \\
&=: \mathfrak{E}_1 + \mathfrak{E}_2.
\end{align*}

For the first term, using Lemma \ref{lem:CorrelPhi} in the first inequality (with $\Psi=\Phi^{(3)}$), we have that for any $\gamma \in [0,1]$, 
\begin{align*}
& \int_0^1 dt \int_0^t ds \Delta K(t,s) \E\left[ \Phi^{(3)}\left( \cI^\gamma\right) \widetilde{W}^\gamma_s \widetilde{W}^\gamma_t  \right] \\
=& \int_0^1 dt C_{\Phi^{(3)}}(t,t) \int_0^t \Delta K (t,s) ds -  \int_0^1 dt \int_0^t ds \; \Delta K (t,s) \left(C_{\Phi^{(3)}}(t,t)  -C_{\Phi^{(3)}}(t,s)  \right) \\
  \leq &\int_0^1 dt C_{\Phi^{(3)}}(t,t) \int_0^t \Delta K (t,s) ds   \\
  &+   \int_0^1 dt \int_0^t |\Delta K (t,s)| \left( O(|t-s|^{2H} + O(n^{-2H}) + O(\ce(t)^2) + O((n^{-H}+\ce(t)) \ce(s)) \right) ds \\
 = &O(n^{-1}) + O(n^{-(3H+1/2)}),
\end{align*}
where we have used \eqref{eq:ineq2} and \eqref{eq:ineq1}, and it follows that $\mathfrak{E}_1 = O\left(n^{-(3H+1/2)\wedge 1}\right)$.

(The case $k=3$ is slightly different, since $\Phi^{(3)}$ is only continuous and we cannot use the second inequality in Lemma \ref{lem:CorrelPhi}.  In that case, we only obtain
\[  \mathfrak{E}_1 \lesssim  \int_0^1 dt \int_0^t ds\left| \Delta K(t,s) \right| (1+ \ce(t))(1+\ce(s)) \lesssim n^{-H-1/2} \]
which is sufficient since it is still smaller than $O(n^{-3H})$.)

For the second term, we rewrite it as
\begin{align*}
 \mathfrak{E}_2 &= \alpha \int_0^1 \int_0^t dt\; ds\; \Delta (K^2)(t,s) + \beta \int_0^1\int_0^t dt \;ds\; (\Delta K(t,s))^2 \\
&= \alpha O(n^{-1}) + \beta O(n^{-2H}),
\end{align*}
using \eqref{eq:strong} and \eqref{eq:wrquad}, where
\[\alpha = \frac{1}{2}\int_\theta^{1-\theta} d\gamma\;  \E\left[ \Phi''\left(\cI^\gamma \right) \right] = O(1)\]
and
\begin{align*}
 \beta &= \int_{\theta}^{1-\theta} d\gamma\; (\gamma - \frac{1}{2})  \E\left[ \Phi''\left( \cI^\gamma \right)  \right] \\
&= \int_{\theta}^{1/2}d \gamma\;   (\gamma - \frac{1}{2})  \left( \E \left[ \Phi''\left( \cI^\gamma \right)\right] - \E \left[ \Phi''\left( \cI^{1-\gamma} \right)\right]\right).
\end{align*}

By the induction hypothesis, the integrand is $O\left(n^{-(3H+1/2)\wedge 1}\right) + O\left(n^{-(k-2)H})\right)$, uniformly over $\gamma \in [0,1/2]$, and we can conclude.

\end{proof}
\subsection{Proof of Theorem \ref{thm:main}, case (2)}

We keep the same notations as in the previous subsections, and note that $D_{s} f(\widehat{W}_t) = f'(\widehat{W}_t) K(t,s)$.

Then we have (using It\^o's formula in the first equality, and Malliavin integration by parts in the second)
\begin{align*}
 \E \left[ \left( \int_0^1 f(\widehat{W}_t)d{W}_t \right)^3 \right]  &= 3 \int_0^1 dt \; \E \left[ \left(\int_0^t f(\widehat{W}_s) dW_s\right) f(\widehat{W}_t)^2 \right]  \\
&=  6 \int_0^1 dt \int_0^t ds\; \E \left[  f(\widehat{W}_s) (ff')(\widehat{W}_t)  \right] K(t,s).
\end{align*}

The same computation holds if $K$ is replaced by $K'$, and we deduce that for $\Phi(x)=x^3$, the weak error is estimated by
\[ \cE_{x^3} \lesssim \cE^{(1)} +  \cE^{(2)}, \]
with
\[  \cE^{(1)} = \int_0^1 dt \int_0^t  ds\;  \phi_f(t,s)  \Delta K(t,s), \]
and 
\[  \cE^{(2)} =  \int_0^1 dt \int_0^t ds \;\left(\phi_f(t,s) - \phi_f(\eta(t),\eta(s)) \right) K'(t,s), \]
where
\[ \phi_f(t,s) := \E \left[ f(\widehat{W}_s) (ff')(\widehat{W}_t) \right]. \]

We then state the following lemma, whose proof is relegated to section \ref{subsec:42}.

\begin{lemma} \label{lem:Phi}
Let $\phi(t,s) = \E \left[ \psi(\widehat{W}_s, \widehat{W}_t)\right]$ for $\psi:\R^2 \to \R$ such that $\partial_1 \psi, \partial^2_1 \psi$ and $\partial_2 \psi$ are bounded. It then holds that for all $s\leq t$ in $[0,1]$,
\begin{equation} \label{eq:ineqphi1}
 \left| \phi(t,t) - \phi(t,s)  \right| +   \left| \phi(s,s) - \phi(t,s)  \right| \lesssim (t-s)^{2H} t^{-H}, 
\end{equation}
\begin{equation}\label{eq:ineqphi2}
 \left| \partial_s\phi(t,s) \right|   \lesssim n^{-1} \left( (t-s)^{2H-1} + s^{2H-1}\right) ,
 \end{equation}
\begin{equation}\label{eq:ineqphi3}
 \left| \partial_s\phi(t,s) \right|   \lesssim n^{-1} \left( (t-s)^{2H-1} + s^{2H-1}\right) ,
 \end{equation}
\end{lemma}

We can then finish the proof of the Theorem. First we note that, using \eqref{eq:ineqphi1},
\begin{align*}
 \cE^{(1)} \lesssim \int_0^1 dt\; \phi_f(t,t) \left(\int_0^t \Delta K(t,s) ds \right) + \int_0^1 dt \; t^{-H} \int_0^t  ds\; |\Delta K(t,s) |O(|t-s|^{2H}) \lesssim n^{-1} + n^{-3H-1/2}.
\end{align*}
We then estimate $\cE^{(2)}$ by splitting the integration domain : 
\begin{align*}
 \cE^{(2)} & = \int_{0 \leq s \leq t \leq 1} dt \; ds \; \left(\phi_f(t,s) - \phi_f(\eta(t),\eta(s)) \right) K'(t,s) \\
&= \int_{0 \leq s \leq t \leq 1, t-s \leq 2 n^{-1}, s \geq 2n^{-1}} (...) + \int_{0 \leq s \leq t \leq 1, t-s \geq 2 n^{-1}, s \geq 2n^{-1}} (...) + \int_{0 \leq s \leq t \leq 1, s \leq 2n^{-1}} (...) \\
&=: \cE^{(2)}_1 + \cE^{(2)}_2 + \cE^{(2)}_3. 
\end{align*}

First, note that
\eqref{eq:ineqphi1} implies that,  if $2n^{-1} \leq s \leq t \leq s +2n^{-1}$, 
\[ \;\; \phi(t,s)- \phi(\eta(t),\eta(s)) \; \lesssim  \; n^{-2H} ( s \wedge \eta(t))^{-H}\; \lesssim\;  n^{-2H} t^{-H} \]
This yields that
\begin{align*}
\cE^{(2)}_1 &\lesssim n^{-2H}  \int_{0 \leq s \leq t \leq 1, \; t-s \leq 2n^{-1}} dt\;ds \; t^{-H} K'(t,s) \\
& \lesssim n^{-2H} \int_0^1 dt \; t^{-H} \left(\int_{0}^{2n^{-1}} v^{H-1/2} dv \right) \\
&\lesssim n^{-3H - 1/2}.
\end{align*}

Second, using \eqref{eq:ineqphi2}-\eqref{eq:ineqphi3}, we see that for $t-s \geq 2n^{-1}$, $s\geq 2n^{-1}$, 
\[  \phi(t,s) - \phi(\eta(t),\eta(s)) \lesssim n^{-1} t^{-H} \left( (t-s)^{2H-1} + s^{2H-1} \right),  \]
and this yields 
\begin{align*}
\cE^{(2)}_2 &\lesssim \int  \int_{t-s \geq 2 n^{-1}, s \geq 2 n^{-1}} dt \;ds \;K(t,s) \left( (t-s)^{2H-1}+s^{2H-1} \right) n^{-1} \\
 &= n^{-1} \int_{0}^1 ds \int_{(s+2n^{-1})\wedge 1}^{1} dt\;  (t-s)^{3H-3/2} + n^{-1} \int_{0}^1 ds s^{2H-1} \int_{(s+2n^{-1})\wedge 1}^{1} dt\;  (t-s)^{H-1/2} \\
 &\lesssim n^{-1} n^{(3H-1/2) \vee 0} + n^{-1} \\
& \lesssim  n^{-(3H+1/2)\wedge 1}.
\end{align*}

Finally, it remains to estimate

\[ \cE^{(2)}_3 \lesssim  \int_0^1 dt \int_{0 \leq s \leq 2n^{-1}}ds\; K'(t,s) \lesssim \int_0^1 dt  \left(\eta(t)^{H+1/2} - (\eta(t)-2n^{-1})_+^{H+1/2} \right) \lesssim n^{-1}.\]

\section{Weak error for the hybrid scheme} \label{sec:3}

In this section, we consider the approximation induced by the hybrid scheme, which was introduced in \cite{BLP15}.

We keep the notations from the previous section, in particular write $K(t,s) = (t-s)_+^{H-1/2} = k(t-s)$, fix a positive integer $n$ and let $h=\frac{1}{n}$. For the hybrid scheme, the kernel $K'$ can be written as 
\[ K'(t,s) = \widecheck{K}(\eta(t),s)= \check{k}(\eta(t)-s),\]
 where, for a fixed integer $\kappa$, 
\begin{equation}
\check{k}(r) = k(r) \mbox{ for } 0 < r < \kappa h, 
\end{equation}
\begin{equation}
\check{k}(r) =  \check{k}_{\ell} \mbox{ for } \ell h \leq r < (\ell+1) h, \; \kappa \leq \ell \leq (n-1),
\end{equation}
where the $\check{k}_{\ell}$ are given weights.

We then consider
\[ \widecheck{W}_t = \widecheck{W}_{\eta(t)} =  \int_0^{\eta(t)}  \widecheck{K}(\eta(t),s) dW_s ,\]
\[ \widecheck{\cI} = \int_0^1 f\left( \widecheck{W}_t\right) dW_t \]
and will be interested in
\[ \widecheck{\cE}_{\Phi} =  \E\left[ \Phi(\cI)\right] -  \E\left[ \Phi(\ccI)\right].\]

Let us discuss the choice of the weights $\check{k}_{\ell}$. Classical choices proposed in the literature are  e.g
\begin{itemize}
\item $\check{k}_{\ell} = k((\ell+1)h)$ (left-point)
\item $\check{k}_{\ell}= k((\ell+1/2)h)$ (mid-point)
\item $\check{k}_{\ell} = \frac{1}{h} \int_{ \ell h}^{(\ell+1) h} k $. (it is shown in \cite{BLP15} that this choice minimizes the mean square error (MSE) between $\cI$ and $\widecheck{\cI}$).
\end{itemize}
However, for the above choices, the weak error cannot be of better order than $2H$, as can be seen by considering quadratics. For instance, for the MSE minimizing weights,
\[ \E[ \widehat{W}_1^2] - \E[ \widecheck{W}_1^2]   \sim_{n \to \infty}  C n^{-2H}\]
where
\[ C =  \sum_{i=\kappa}^{\infty} \left( \int_i^{i+1} k^2 - (\int_i^{i+1} k)^2\right) >0 \]
(this uses the self-similarity of $k$). A similar formula holds true for all the other grid-points, and this yields (for $f(x)=x$), $\left|\E[\cI^2] - \E[\ccI^2] \right|\gtrsim n^{-2H}$.

The good choice in our context is therefore to choose the weights that match the second moment, namely :
\begin{equation} \label{eq:MMweights}
\check{k}_{\ell} =  \left( \frac{1}{h} \int_{ \ell h}^{(\ell+1) h} k^2(r) dr \right)^{1/2} 
\end{equation} 
 Then it holds that $\int K(\eta(t),r)^2 dr = \int \widecheck{K}(\eta(t), r)^2 dr$ for each $t$, so that the second moment $\E[\ccI^2] $ coincides with $\E[(\cI')^2]$ obtained by exact discretization, and the weak rate for quadratics is $1$. (These weights have been first suggested in \cite{HJM17}).
   
  We illustrate these considerations in Figure \ref{fig:quadratic hybrid}, where we plot $\ccE_{x^2}$ as a function of $n$ for the four choice of weights described above, for $H=0.02$. We see that, for the first three choices, as expected, the eventual decrease becomes very slow. In addition, while the left-point weights always give much worse results, for small values of $n$ mid-point or MSE minimizing weights give comparable results to the moment matching ones (this is due to the fact that while the order of convergence in $n^{-2H}$ is the same in these three cases, the multiplying  constant will be significantly smaller for the second and third choices).  
  
 \begin{figure}[!h]
\includegraphics[width=400bp]{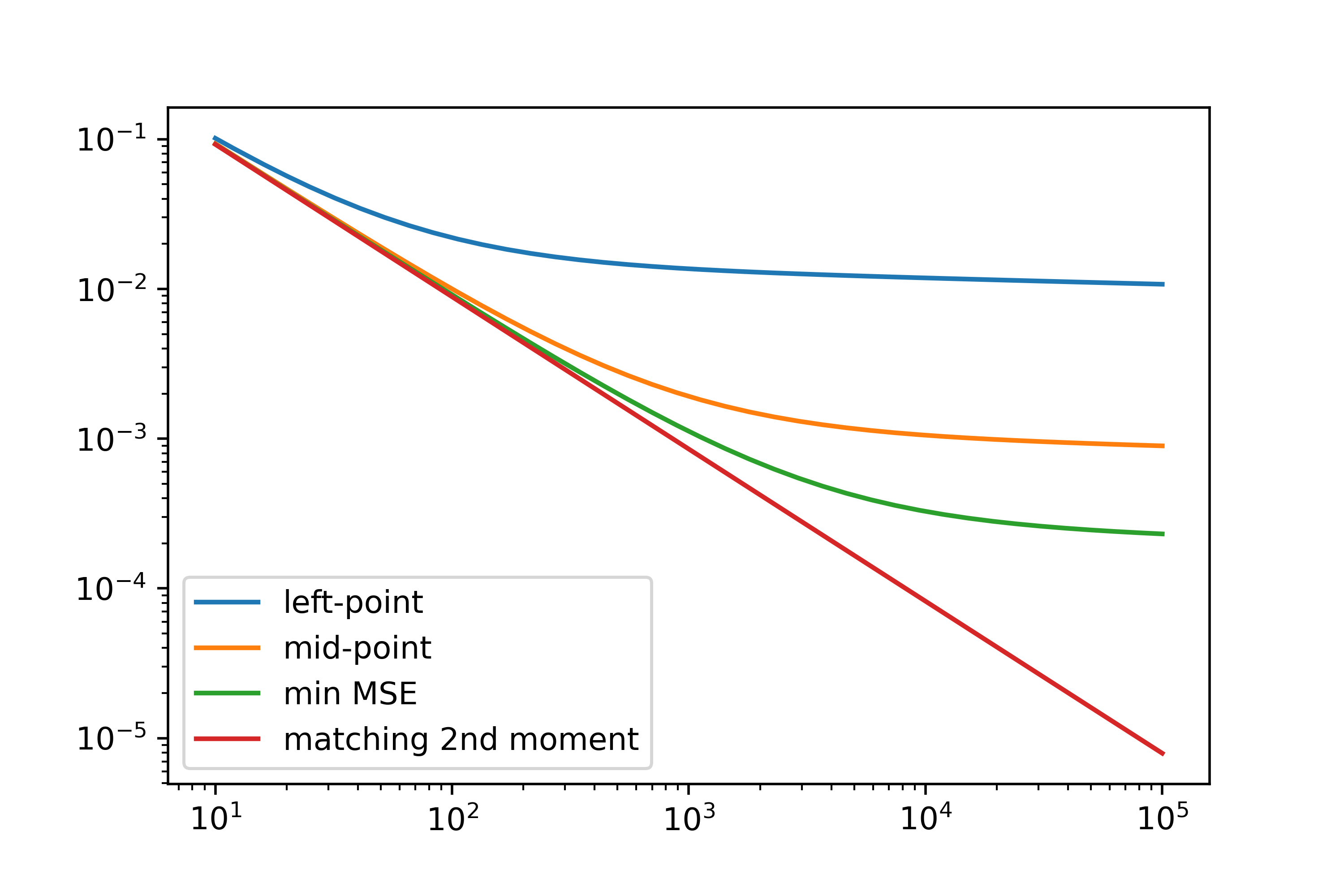}
\caption{Plot of $\cE_{x^2}$ as a function of $n$ for various choices of weights and $H=0.02$, $\kappa=1$.}
\label{fig:quadratic hybrid}
\end{figure}

Let us formalize the properties of $\widecheck{K}$ (with weights  chosen as in \eqref{eq:MMweights}) that we will use.

\begin{assumption} \label{asn:cK}
The kernel $\widecheck{K}$ satisfies $\widecheck{K}(t,s) =  \widecheck{K}(\eta(t),s)$ for all $s \leq t \in [0,1]$, and in addition, for all grid-points $t=\eta(t)$,
\begin{equation} \label{eq:matching}
\int_0^t \widecheck{K}(t,s)^2 ds = \int_0^t K(t,s)^2 ds,
\end{equation}

\begin{equation} \label{eq:cK2}
\forall s\leq t,  \;\;\;\; \widecheck{K}(t,s) = K(t,\check{s}), \;\;\;\mbox{ where }\check{s} = \check{s}(t,s) \mbox{ satisfies }|\check{s}-s|\leq h  \mbox{ always, and }\check{s}=s \mbox{ for }t-s\leq \kappa h,
\end{equation}
where $\kappa \geq 1$ is fixed.
\end{assumption}

We now  give the main result of this section, which states a weak error rate of $H+\frac 1 2$ (for the same special cases as in Theorem \ref{thm:main}).

\begin{theorem} \label{thm:hybrid}
Let $\widecheck{W}_t = \int_0^t \widecheck{K}(t,s) dW_s$, where $\widecheck{K}$ satisfies Assumption \ref{asn:cK}. Further assume that either~:

(1) $f(x)=x$ and $\Phi$ is a $C^{2\lceil1/4H \rceil+1}_b$ function,

or

(2) $f$ is $C^3_b$ and $\Phi$ is a cubic polynomial.

Then there exists a constant $C$, which does not depend on $n$, such that 
\[ \left|\ccE_{\Phi}\right|  \leq C \left(\frac{1}{n}\right)^{H+1/2 }. \]

 \end{theorem}

\begin{remark}  \label{rmk:constant}
(1) Let us explain the difference in rates between hybrid scheme and exact discretization. It comes from the fact that it is not possible to choose the weights to match both second and first moments of the kernel. The second moment has to be matched (in order to avoid rate $2H$), but then \eqref{eq:ineq2} does not hold. It is indeed easy to check that one has instead
\[ \int_{0}^t \Delta K(t,s) ds  \sim_{n \to \infty} n^{-H-1/2} t^{H+1/2} C_{\kappa,H},\]
with
\[ C_{\kappa,H} = \sum_{k=\kappa}^\infty \left[ \left( \frac{(k + 1)^{2H} - k^{2H}}{2H} \right)^{1/2} - \frac{(k+1)^{H+1/2} - k^{H+1/2}}{H+1/2}   \right]. \]
Inspecting the proof of Theorem \ref{thm:main} (1), we can then expect a leading order term $\widehat{C}^\Phi_{\kappa,H} n^{-H-1/2}$, with
\begin{equation} \label{eq:widehatC}
 \widehat{C}^\Phi_{\kappa,H} = C_{\kappa,H} \int_0^1 \E\left[ \Phi^{(3)}(\cI) \widehat{W}_t^2 \right] t^{H+1/2} dt.
 \end{equation}

(2) We see from Theorems \ref{thm:main} and \ref{thm:hybrid} that the Cholesky scheme has a higher (weak) convergence rate than the hybrid scheme. However, it also has a higher computational cost ($O(n^2)$ vs $O(n \log(n))$).  Given an error tolerance level of order $\varepsilon$, we can compute the required computational costs for both schemes and obtain (ignoring logarithmic terms)
\[Cost_{Hybrid} = O(\varepsilon^{- \frac{1}{H+1/2}}) , \;\;\;\; Cost_{Cholesky} =O(\varepsilon^{- \frac{2}{(3H+1/2) \wedge 1)}}).
 \]
This yields in particular, the hybrid scheme is always (asymptotically) less costly, for any $H \in (0,1/2)$.

(3) In fact, it is not clear that the difference in the asymptotic rates between Theorems \ref{thm:main} and \ref{thm:hybrid} is  relevant in practice. Indeed, the constant $\widehat{C}^\Phi_{\kappa,H}$ from \eqref{eq:widehatC} is typically rather small
, at least compared to the loss of $n^{-2H}$. For example, for $\kappa=1$, $H=0.1$ and $\Phi(x)=x^3/6$ one has $\widehat{C}^\Phi_{\kappa,H}\approx 0.012$. Then one can check that $\widehat{C}^\Phi_{\kappa,H} n^{-H-1/2}$ is only bigger than the error of the Cholesky scheme $\approx \widetilde{C} n^{-3H-1/2}$ ($\widetilde{C} \approx 3$ being estimated from numerical values) for $n$ of order $10^{12}$, which is much higher than the discretization sizes used in practical situations.
\end{remark}

\begin{remark} 
The kernel of the hybrid scheme is piecewise constant (away from the singularity), but our Assumption \ref{asn:cK} covers more general approximations. In particular, our result also applies to Fukasawa and Hirano's 3R scheme \cite{FH21}, where the chosen approximation is of the form
\[ k(r) \approx \alpha_k +  \beta_k (r  -  (k-\kappa+1) h)^{H-1/2}, \;\;\;\;\; r \in [kh, (k+1) h),\;\;\;\;\; k \geq \kappa.\]
Our result then gives an asymptotic weak rate $H+\frac{1}{2}$ in this case as well, assuming that the $\alpha_k$ and $\beta_k$ are chosen to match the second moments of $k$ (and satisfy the technical condition \eqref{eq:cK2}). (Note that since the focus of \cite{FH21} was on reducing the mean square error, their weights were chosen differently). 

We can also expect that the method of proof can be applied to even more general approximations for the fractional kernel, such as the multi-factor  Markovian approximations \cite{CC98} (however in that case \eqref{eq:cK2} cannot be satisfied, so that some arguments would need to be modified). 
\end{remark}

In Figure \ref{fig:cubic hybrid} below, we plot the error for $\Phi(x)=x^3/6$. This confirms the observation made in Remark \ref{rmk:constant} (3), in that in that case, for realistic step-sizes, the hybrid scheme (with weights matching the 2nd moment) gives results which are indistiguishable from Cholesky discretization.

\begin{figure}[!h]
\includegraphics[width=400bp]{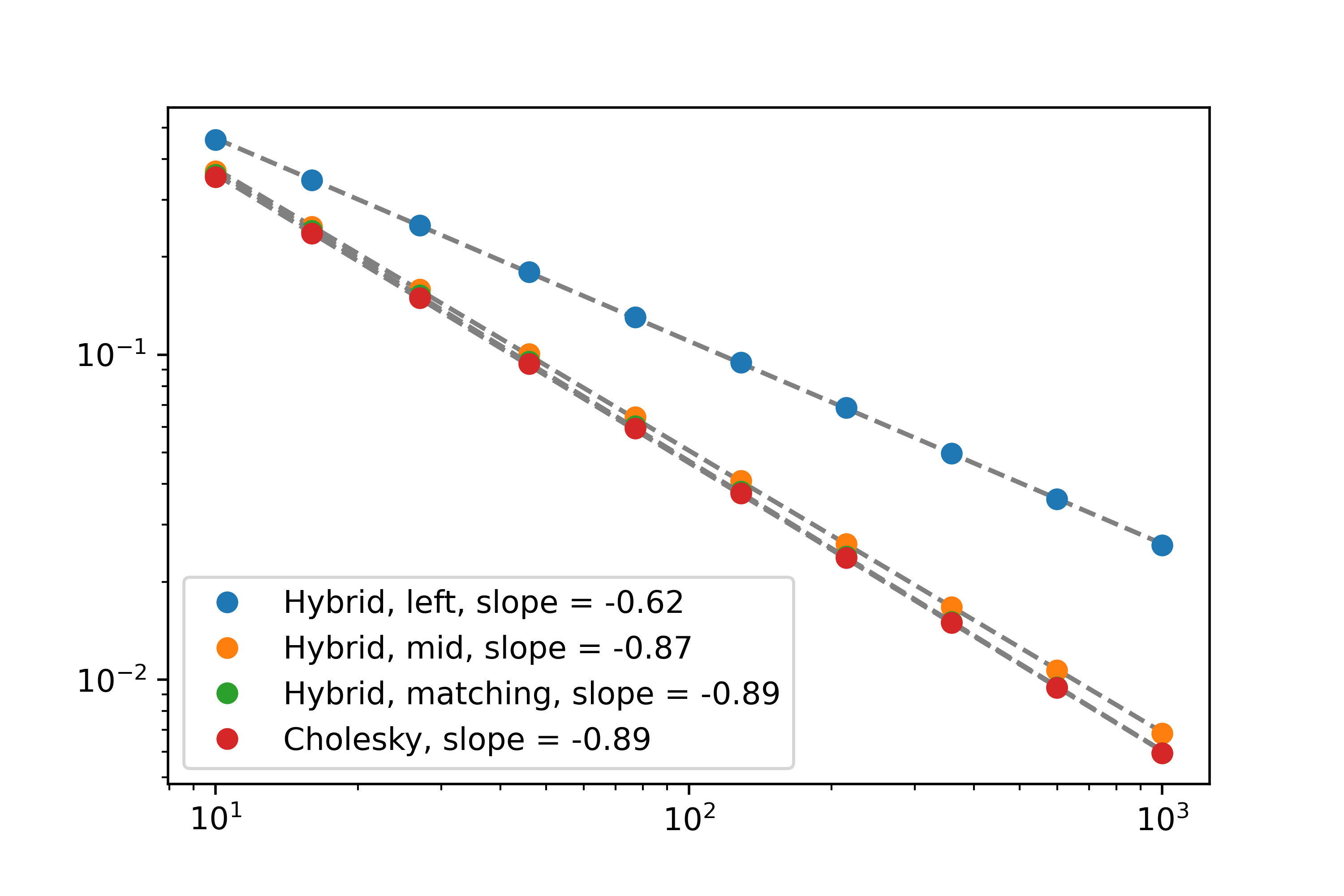}
\caption{Plot of $\cE_{x^3/6}$ when $f(x)=x$ and $H=0.15$ as a function of $n$, for Cholesky and Hybrid ($\kappa=1$) schemes with various choices of weights. The dotted lines are linear regressions, with slopes indicated in the legend.}
\label{fig:cubic hybrid}
\end{figure}

\subsection{Proof of Theorem \ref{thm:hybrid}}

\subsubsection{Proof of case (1)}

We keep the same notations as in the proof of Theorem \ref{thm:main} (1), with now $K'(t,s) = \widecheck{K}(\eta(t),s)$. The proof is essentially the same and in fact, simpler, since we only need to obtain rate $H+\frac{1}{2}$, so we only sketch it. The first three inequalities in Lemma \ref{lem:inequalities} still hold. From Lemma \ref{lem:CorrelPhi} we now only need the fact that
\[ \E\left[ (\widetilde{W}_t^\theta)^2 \right] \lesssim 1 +\ce(t)^2, \mbox{ where } \int_0^1 dt  \int_0^t ds \;(1+ \ce(t))(1+\ce(s)) |\Delta K(t,s)| \lesssim n^{-2H}\]
which is again proved similarly (the function $\ce(t)$ is the same).

In the induction step, we now prove a bound of order $O\left(n^{-H-1/2}\right) + O\left(n^{-kH}\right)$ for $C^k_b$ test functions $\Phi$. This follows from estimating $\mathfrak{E}_1$ by
\begin{align*}
& \int_0^1 dt \int_0^t ds \; \Delta K(t,s) \E\left[ \Phi^{(3)}\left( \cI^\gamma\right) \widetilde{W}^\gamma_s \widetilde{W}^\gamma_t  \right] \\
\lesssim& \int_0^1 dt\; (1+ \ce(t))\int_0^t ds\; (1+\ce(s)) \left|\Delta K (t,s)\right|  \\
\lesssim &\; n^{-H-1/2}.
\end{align*}
while the estimate for $\mathfrak{E}_2$, which uses \eqref{eq:strong} and \eqref{eq:wrquad}, remains the same.

\subsubsection{ Proof of case (2)}

Let $\Phi(x)=x^3$. By the result of Theorem \ref{thm:main} (2), it is enough to compare $\E [(\ccI)^3] - \E [(\cI')^3]$, where $\cI' =\int_0^1 f(\widehat{W}_{\eta(t)}) dW_t$. 
By the same argument as in the proof of Theorem \ref{thm:main} (2), we write this difference as $\cE^{(1)}$ + $\cE^{(2)}$, where
\[ \cE^{(1)} = \int_0^1 dt \int_0^t ds\; \check{\phi}_f(\eta(t),\eta(s)) \left( K(\eta(t),s) - \widecheck{K}(\eta(t),s) \right)  \]
and 
\[ \cE^{(2)} =  \int_0^1 dt \int_0^t ds \;\left( \phi_f(\eta(t),\eta(s)) - \check{\phi}_f(\eta(t),\eta(s)) \right) K(t,s) , \]
where
\[ \phi_f(t,s) = \E \left[ f(\widehat{W}_s) (ff')(\widehat{W}_t) \right], \;\;\; \check{\phi}_f(t,s) = \E \left[ f(\widecheck{W}_s) (ff')(\widecheck{W}_t) \right]. \]

The first term is simple, since we can used boundedness of $\phi_f$ to obtain
\begin{align*}
 \cE^{(1)}& \lesssim \int_0^1 dt \int_0^t ds \; \left| K(\eta(t),s) - \widecheck{K}(\eta(t),s) \right|  \\
 & \lesssim \int_0^1 dt \;\int_0^{(\eta(t)-\kappa h)_+}  n^{-1} (t-s)^{H-3/2} ds \\
 &\lesssim \; n^{- H-1/2}.
 \end{align*}
 
For the second term, we use the following lemma (the proof being deferred to the Appendix).
 
\begin{lemma} \label{lem:PhiHybrid}
Let $\phi(t,s) = \E \left[ \psi(\widehat{W}_s, \widehat{W}_t)\right]$,   $\check{\phi}(t,s) = \E \left[ \psi(\widecheck{W}_s, \widecheck{W}_t)\right]$ for $\psi:\R^2 \to \R$ such that $\psi, \partial_1 \psi, \partial^2_1 \psi$ and $\partial_2 \psi$ are bounded, where $\widehat{W}$ is the Riemann-Liouville fBm, and $ \widecheck{W}_t =\int_0^t \widecheck{K}(t,s) dW_s$ where $\widecheck{K}$ satisfies Assumption \ref{asn:cK}. It then holds that for all $s\leq t$ in $[0,1]$ such that $t=\eta(t)$, $s=\eta(s)$ and $t-s \geq \kappa h$,
\begin{equation} \label{eq:ineqphiH}
 \left| \phi(t,s) - \check{\phi}(t,s)  \right| \lesssim n^{-1} (t-s)^{2H-1} t^{-H} + n^{-H-1/2} (t-s)^{H-1/2} t^{-H} . 
\end{equation}
 \end{lemma}
 
 We continue with the proof of the Theorem. The contribution to $\cE^{(2)}$ of close points $t$ and $s$ is simple to bound, using boundedness of $\phi$ and $\check{\phi}$, since
 \[ \int_{0 \leq t-s \leq (\kappa +1) h}  dt \;ds \;  (t-s)^{H-1/2} \lesssim n^{-H-1/2}. \]
 It remains to estimate the contribution of $t,s$ with $t-s \geq (\kappa +1) h$, for which we use Lemma \ref{lem:PhiHybrid} to obtain the bound
 \begin{align*}
  & \;\;\; \int_0^1 dt \; \int_0^{(t- (\kappa+1) h)_+} ds \left( \phi_f(\eta(t),\eta(s)) - \check{\phi}_f(\eta(t),\eta(s)) \right)  (t-s)^{H-1/2} \\
&\lesssim \int_0^1 dt \; t^{-H} \; n^{-1} \int_0^{(t- (\kappa+1) h)_+} ds \;(t-s)^{3H-3/2}   + \int_0^1 dt \;t^{-H} \;  n^{-H-1/2}\int_0^{(t- (\kappa+1) h)_+} ds \; (t-s)^{2H-1} \\
&\lesssim \;\; \left(\int_0^1 dt  \; t^{-H}\right) n^{-1} n^{(-3H+1/2) \vee 0}  \; + \left(\int_0^1 dt  \; t^{-H}\right) n^{-H-1/2}  \\
& \lesssim \;\;\;\;\;\;\;\;\;\; n^{- (3H+1/2) \wedge 1}  + n^{-H-1/2}.
 \end{align*}

\section{Technical proofs} \label{sec:4}

\subsection{Proof of Lemma \ref{lem:CorrelPhi}} \label{subsec:41}

We let $ \ep(t) = \ \lceil nt \rceil/n  \geq t$ and define
\begin{equation} \label{eq:defE}
\ce(t) = h^{1/2} \left( \ep(t) - t \right)^{H-1/2}
\end{equation}

\begin{lemma}
It holds that
\begin{equation*}
\int_0^1 dt \;\ce(t)^2 dt \;\lesssim n^{-2H},
\end{equation*}
and 
\begin{equation*}
\int_0^1 dt \int_0^t ds\; |\Delta K(t,s)| \ce(t) \ce(s) \lesssim n^{-3H-1/2}.
\end{equation*}
\end{lemma}

\begin{proof}
The first assertion is immediate. We now prove the second assertion. Since $\int_0^1 \ce(t) dt \lesssim n^{-H}$, it suffices to show that
\[ \sup_{0\leq t \leq 1} \int_0^t |\Delta K(t,s)| \ce(s) ds \lesssim n^{-2H-1/2}.\]
 We separate the integral in two terms depending on whether $t-s\geq h$ or $t-s\leq h$. The first term is
\begin{align*}
\int_0^{t-h} ds \;|\Delta K(t,s)|  \ce(s) &\lesssim h   \int_0^{t-h} (t-s)^{H-3/2}\ce(s) ds  \\
&\leq  h  \sum_{k\geq 1} (kh)^{H-3/2} (\int_0^h \ce(s) ds) \\
&\lesssim h^{2H+1/2} ,
\end{align*}
and the second one is bounded by
\begin{align*}
\int_{t-h}^t ds  \;\left( (t-s)^{H-1/2} + (\eta(t)-s)_+^{H-1/2} \right)\left((\eta(t)-s)_+^{H-1/2} + (\ep(t)-s)_+^{H-1/2}\right) h^{1/2}  \lesssim h^{2H+1/2}.
\end{align*}
\end{proof}

We consider the case $\theta=0$, i.e. $\widetilde{W} = \widehat{W} + \overline{W}$  which has corresponding kernel $\widetilde{K}(s,t)= K(t,s) + K(s,t) = |s-t|^{H-1/2}$. Similarly define $\Delta \widetilde{W}$, $\Delta \widetilde{K}$. We then have the following estimates.

\begin{lemma} \label{lem:Cor0}
It holds that for all $0 \leq s \leq t$,
\begin{equation} \label{eq:Correl0}
\E \left[ \widetilde{W}_t \left( \widetilde{W}_t - \widetilde{W}_s \right) \right] \lesssim (t-s)^{2H},
\end{equation}
\begin{equation} \label{eq:CorrelDelta}
\E \left[ \Delta\widetilde{W}_t^2 \right] \lesssim n^{-2H} + \ce(t)^2,
\end{equation}
and for all $s\leq t$,
\begin{equation} \label{eq:CorrelDelta2}
\E \left[ \widetilde{W}_t \left( \Delta \widetilde{W}_t -  \Delta \widetilde{W}_s \right) \right]  \lesssim h^{2H} + |t-s|^{2H} +  h^{H} \left(\ce(t) + \ce(s)\right).
\end{equation}
\end{lemma}

\begin{proof}
\eqref{eq:Correl0} follows from
\[\int_0^1 |t-r|^{H-1/2} \left(|t-r|^{H-1/2} - |s-r|^{H-1/2} \right) dr \leq (t-s)^{2H} \int_{-\infty}^\infty |1-u|^{H-1/2}  \left(|1-u|^{H-1/2} - |u|^{H-1/2} \right) du \]
using the change of variables $r = (t-s) u$.

We now show \eqref{eq:CorrelDelta}. It holds that
\begin{align*}
\E \left[ \Delta\widetilde{W}_t^2 \right] = \int_0^1 dr  \left(|t-r|^{H-1/2} -(\eta(t)-r)^{H-1/2}_+ - (\eta(r)-t)^{H-1/2}_+\right)^2
\end{align*}
We then split this integral into two contributions.

On the points where $|r-t|\geq 2h$, the integrand is bounded by a multiple of $h^2 |t-r|^{2H-3}$, for a total contribution of order $h^{2H}$.

For the points where $|r-t|\leq 2h$, we bound the integral by a multiple of
\begin{align*}
\int_{t-2h}^{t+2h}  dr  \left( |t-r|^{2H-1} + | \eta(t)-r|^{2H-1} \right)+ \int_{t}^{t+2h} (\eta(r)-t)^{2H-1}_+ dr \lesssim h^{2H} + h  \left( \ep(t) - t \right)^{2H-1}.
\end{align*}

We now treat \eqref{eq:CorrelDelta2}, for which one has
\begin{align*}
\E \left[ \widetilde{W}_t \left( \Delta \widetilde{W}_t -  \Delta \widetilde{W}_s \right) \right] & = \int_0^1 dr \; |t-r|^{H-1/2} \left(\Delta \widetilde{K}(t,r) - \Delta \widetilde{K}(s,r)\right) dr \\
&\lesssim \int_{|t-r| \leq 2h}dr \; |t-r|^{H-1/2} \left(\widetilde{K}(t,r) + \widetilde{K}'(t,r) + \widetilde{K}(s,r) + \widetilde{K}'(s,r) \right) \\
&+  \int_{|s-r| \leq 2h} dr\; |t-r|^{H-1/2} \left(\widetilde{K}(t,r) + \widetilde{K}'(t,r) + \widetilde{K}(s,r) + \widetilde{K}'(s,r) \right) \\
&+  \int_{r \geq t+ 2h,}dr\; |t-r|^{H-1/2}\left( |s-t| + h\right) |t-r|^{H-3/2}   \\
&+  \int_{r \leq s- 2h,}dr\;  |t-r|^{H-1/2}\left( |s-t| + h\right) |s-r|^{H-3/2}   \\
&+ \int_{s +h \leq r\leq t-h} dr \; |t-r|^{H-1/2} \left( h |t-r|^{H-3/2}  + h  |s-r|^{H-3/2}\right)
\end{align*}
The first two lines are bounded by $h^{2H} + h^{H} \left(\ce(t) + \ce(s)\right)$,  the third and fourth lines bounded by $h^{2H} + h^{2H-1} |t-s| \lesssim h^{2H} + |t-s|^{2H}$, and the last line by $h^{2H}$.
\end{proof}

\begin{lemma} \label{lem:Cor1}
For any $\theta \in [0,1]$, letting
\begin{equation*}
C_{1}^{\theta} : (s,t) \mapsto \E\left[\widetilde{W}^{\theta}_s  \widetilde{W}^{\theta}_t\right]
\end{equation*}
it holds that
\begin{equation} \label{eq:C1t}
C_1^{\theta}(t,t) \lesssim 1+ \cE(t)^2,
\end{equation}
and
\begin{equation} \label{eq:C1ts}
C_1^{\theta}(t,t) - C_{1}^{\theta}(t,s) \lesssim (t-s)^{2H} + n^{-2H} + \ce(t)^2 + (n^{-H}+ \ce(t))\ce(s).
\end{equation}
\end{lemma}

\begin{proof}
This is immediate from Lemma \ref{lem:Cor0} and the Cauchy-Schwarz inequality, writing $\widetilde{W}^\theta = \widetilde{W} + \theta \Delta W$.
%
  \end{proof}

\begin{lemma} \label{lem:CorrelPhiApp}
For any $\theta \in [0,1]$, and any $C^1_b$ function $\Psi$, the map
\begin{equation*}
C_{\Psi}^\theta : (s,t) \mapsto \E\left[\Psi(\cI^\theta)\widetilde{W}^\theta_s \widetilde{W}^\theta_{t} \right]
\end{equation*}
satisfies the same as above, namely
\begin{equation*}
C_{\Psi}^{\theta}(t,t) \lesssim 1+ \ce(t)^2,
\end{equation*}
and
\begin{equation*}
C_{\Psi}^{\theta}(t,t) - C_{1}^{\theta}(t,s) \lesssim (t-s)^{2H} + n^{-2H} + \ce(t)^2 + \ce(t)\ce(s).
\end{equation*}

\end{lemma}

\begin{proof}
The first inequality is immediate since $\Psi$ is bounded.

For the second one, we write
\begin{align*}
C_{\Psi}^{\theta}(s,t)& =\E\left[\Psi(\cI^\theta)  \widetilde{W}^{\theta}_s \widetilde{W}^{\theta}_t\right] \\
&= \E \left[ \int_0^1 D_r \left( \Psi(\cI^{\theta}) \widetilde{W}_t \right) \widetilde{K}^{\theta}(s,r) dr \right] \\
&=  \E [\Psi(\cI^{\theta}) ] \left( \int_0^1 \widetilde{K}^{\theta}(t,r) \widetilde{K}^{\theta}(s,r) dr \right) \\
& + \int_0^1 C_{\Phi'}^{\theta}(t,r) \widetilde{K}^{\theta}(s,r) dr.
\end{align*}

The first part is equal to $\E[\Psi(\cI^{\theta})] C_1(s,t)$, and we can therefore apply the results of the previous lemma.

For the second part, we note that by Lemma \ref{lem:Cor0}, it holds that
\[ C_{\Psi'}^\theta(t,r) \lesssim (1+\ce(t)) (1+\ce(r)), \]
so that we need to bound
\[ \int_0^1 (1+\ce(r)) \left|\widetilde{K}^{\theta}(t,r)-\widetilde{K}^{\theta}(s,r) \right| dr. \]
Splitting as in the previous proofs depending whether $|t-r|$, $|s-r|$ is $\geq h$ or $\leq h$, we obtain a bound of order 
\[ |t-s|^{H+1/2} +h^{H+1/2} + h^{1/2} \cE(s) + h^{1/2} \cE(t)\]
(which is negligible before the first term). 
\end{proof}

\subsection{Proof of Lemma \ref{lem:Phi}} \label{subsec:42}

We first record some properties of the covariance function of Riemann-Liouville fBm. 

\begin{proposition}
Let $C(s,t) = \E[\widehat{W}_t \widehat{W}_s ]$. 
Then it holds that 
\[C(t,t) = C t^{2H} \]
and for $0 \leq s\leq t \leq 1$,
\begin{equation} \label{eq:C1}
\left| C(t,t) - C(s,t) \right| \lesssim (t-s)^{2H}, \;\;\;\; \left| C(t,s) - C(s,s)  \right| \lesssim (t-s)^{2H} 
\end{equation}
\begin{equation} \label{eq:C2}
 \left| \partial_t C(s,t) \right| \lesssim (t-s)^{2H-1} , 
\end{equation}
\begin{equation} \label{eq:C3}
\left| \partial_s C(s,t) \right| \lesssim (t-s)^{2H-1} +  s^{2H-1}.
\end{equation}
\end{proposition}

\begin{proof} It holds that
\begin{equation*}
C(t,t) - C(s,t) = \int_0^t (t-r)^{H-1/2} \left((t-r)^{H-1/2} - (s-r)_+^{H-1/2} \right)dr 
\end{equation*}
Using the change of variables $r = t-(t-s)u$, this leads to
\begin{align*}
C(t,t) - C(s,t) = (t-s)^{2H} \left( \int_{0}^{\frac{t}{t-s}}   u^{H-1/2}  \left(u^{H-1/2} - (1-u)_+^{H-1/2} \right) du\right) \lesssim (t-s)^{2H-1}.
\end{align*}

We also have
\begin{align*}
 \partial_t C(s,t) & = (H-1/2) \int_0^s (s-r)^{H-1/2} (t-r)^{H-3/2} dt \\
 &= (H-1/2) (t-s)^{2H-1} \int_0^{s/(t-s)} u^{H-1/2} (1+u)^{H-3/2}du = O((t-s)^{2H-1}).
 \end{align*}
The bound on $C(t,s) - C(s,s)$ follows.

Finally, writing
\[ C(s,t) = s^{2H} \int_0^1 (t/s-u)^{H-1/2} (1-u)^{H-1/2}  du\]
we obtain
\[ \partial_s C = s^{2H-1} O(1)  + (1/2-H)  s^{2H-2} t  \int_0^1 (t/s-u)^{H-3/2} (1-u)^{H-1/2}du .\]
we then distinguish two cases. First, if $t/s \leq 2$, the integral is $O(t/s-1)^{2H-1}$ and we bound the second term by a multiple of
\[ 
s^{2H-1} (t/s-1)^{2H-1} = (t-s)^{2H-1}. \]
In the case when $t/s \geq 2$, the integral is now $O((t/s)^{H-3/2})$ and this leads to an overal bound of
\[ s^{2H-2} t (t/s)^{H-3/2} = t^{H-1/2} s^{H-1/2} \leq s^{2H-1}. \]
\end{proof}

\begin{lemma} \label{lem:Z12}
Let $\psi:\R^2 \to \R$ be such that $\partial_1 \psi, \partial^2_1 \psi$ and $\partial_2 \psi$ are bounded. Consider
\[
\varphi :(\alpha, \beta, \gamma) \in \R \times \R_+ \times \R \mapsto \E \left[ \psi(\alpha Z_1 + \sqrt{\beta} Z_2, \gamma Z_1)\right]
\]
where $Z_1$, $Z_2$ are independent $\mathcal{N}(0,1)$. Then $\varphi$ is globally Lipschitz.
\end{lemma}

\begin{proof}
Since $\psi$ is Lipschitz, Lipschitz dependence of $\varphi$ in $\alpha$ and $\gamma$ is clear. We then write
\[ \partial_\beta \varphi = \frac{1}{2 \sqrt{\beta}} \E \left[ Z_2 (\partial_1 \psi)(\alpha Z_1 + \sqrt{\beta} Z_2, \gamma Z_1)\right] \]
and note that, since $ \E \left[ Z_2 (\partial_1 \psi)(\alpha Z_1, \gamma Z_1)\right]=0$, the expectation on the r.h.s. is bounded by a multiple of $\sqrt{\beta}$.
\end{proof}

\begin{proof}[Proof of Lemma \ref{lem:Phi}.]

We can use the representation
\[ (\widehat{W_s}, \widehat{W}_t) = \left( \alpha(s,t) Z_1 +  \sqrt{\beta(s,t)} Z_2, \gamma(t) Z_1 \right) \]
where $Z_1$, $Z_2$ are independent $\cN(0,1)$, and
\[ \alpha(s,t)^2 = \frac{C(t,s)^2}{C(t,t)}, \;\;\; \beta(s,t) = C(s,s) - \frac{C(t,s)^2}{C(t,t)}, \;\;\; \gamma(t) = C(t,t)^{1/2}. \]

By Lemma \ref{lem:Z12}, this yields
\begin{align}
\phi(t,t)-\phi(t,s)& \lesssim | \alpha(t,t)- \alpha(t,s) | + | \beta(s,t) - \beta(t,t)| \nonumber \\
&= \sqrt{C(t,t)}- C(t,s)/\sqrt{C(t,t)} + C(s,s) - C(t,s)^2/C(t,t) \nonumber\\
&\lesssim (t-s)^{2H} t^{-H} +  \frac{C(s,s) C(t,t) - C(t,s)^2}{C(t,t)}\nonumber \\
&\lesssim (t-s)^{2H} (t^{-H}+1), \label{eq:phitt}
\end{align}
where we have used \eqref{eq:C1}. A similar computation yields
\begin{equation} \label{eq:phiss}
\phi(s,s)-\phi(t,s) \lesssim (t-s)^{2H} (t^{-H}+1),
\end{equation}

Further using the properties of $C$, it holds that
\[ \partial_s \alpha \lesssim (t-s)^{2H-1} t^{-H} +  s^{2H-1} t^{-H}, \]
\[  \partial_s \beta(s,t) \lesssim  s^{2H-1} + (t-s)^{2H-1} \]
and Lemma \ref{lem:Z12}, we obtain
 \[ 
\partial_s \phi(s,t) \lesssim t^{-H} \left( s^{2H-1} + (t-s)^{2H-1} \right).
\]
Similarly, using that 
\[ \partial_t \alpha \lesssim (t-s)^{2H-1} t^{-H} + t^{H-1} \lesssim (t-s)^{2H-1} t^{-H} \]
\[ \partial_t \beta \lesssim (t-s)^{2H-1} + t^{2H-1} \lesssim (t-s)^{2H-1}, \] 
\[ \partial_t \gamma \lesssim t^{H-1}\lesssim (t-s)^{2H-1} t^{-H} \]
it holds that 
\[ \partial_t \phi \lesssim (t-s)^{2H-1} t^{-H}. \]

\end{proof}

\subsection{Proof of Lemma \ref{lem:PhiHybrid}} \label{subsec:43}

Let $C(s,t) = \E[\widehat{W}_t \widehat{W}_s ]$ and  $\widecheck{C}(s,t) = \E[\widecheck{W}_t \widecheck{W}_s ]$, for grid-points $t$ and $s$ satisfying $t-s \geq \kappa h$.

By Lemma \ref{lem:Z12}, it holds that
\[ \left| \phi(t,s) - \check{\phi}(t,s) \right| \lesssim \left| \alpha(t,s) - \check{\alpha}(t,s) \right| +  \left| \beta(t,s) - \check{\beta}(t,s) \right|, \]
where
\[ \alpha(t,s)= \frac{C(t,s)}{\sqrt{C(t,t)}} , \;\;\;\beta(t,s) = C(s,s) -  \frac{C(t,s)^2}{C(t,t)} \]
and $\check{\alpha}$, $\check{\beta}$ are defined similarly with $\widecheck{C}$ instead of $C$. Note that by assumption, it holds that for all grid-points $t$, $\widecheck{C}(t,t) = C(t,t) = C t^{2H}$. It follows that
\[ \left| \phi(t,s) - \check{\phi}(t,s) \right| \lesssim t^{-H} \left| C(s,t) - \widecheck{C}(s,t)\right|. \]

Now recall that for $t\geq r$,  $\widecheck{K}(t,r) = K(t, \check{r})$ for some $\check{r}$ (depending on $t$ and $r$), with $|r-\check{r}| \leq h$ and similarly if $s\geq r$, $\widecheck{K}(s,r) = K(s, \check{r}')$,  with $|r-\check{r}'| \leq h$, and $\check{r}'=r$ if $s-r \leq \kappa h$.

This yields
\begin{align*}
C(s,t) - \widecheck{C}(s,t)& = \int_0^s (t-r)^{H-1/2} (s-r)^{H-1/2} dr -  \int_0^s (t-\check{r})^{H-1/2} (s-\check{r}')^{H-1/2} dr \\
&\lesssim n^{-1} \int_0^s (t-r)^{H-3/2} (s-r)^{H-1/2} dr  + n^{-1} \int_0^{s- \kappa h} (t-r)^{H-1/2} (s-r)^{H-3/2} dr  \\
&\lesssim n^{-1} (t-s)^{2H-1} + n^{-H-1/2} (t-s)^{H-1/2}
\end{align*}
and the result follows.

\bibliographystyle{alpha}
\bibliography{roughvol}

\end{document}